\documentclass[oribibl,envcountsame,envcountsect]{llncs}

\usepackage[latin1]{inputenc}
\usepackage[T1]{fontenc}

\usepackage{amsmath, amssymb, latexsym, amsfonts}
\usepackage{euscript, stmaryrd, bigstrut}

\usepackage{epic, epsfig, pstricks, pst-node, pst-tree, pst-coil}
\usepackage{gastex}

\usepackage[english]{babel}

\usepackage{verbatim}

\usepackage{ifthen}

\usepackage{url}

\usepackage{listings}

\usepackage{galois}

\usepackage{bbold}

\newboolean{display_full}
\newboolean{display_light}
\setboolean{display_light}{false}
\setboolean{display_full}{false}

\setboolean{display_light}{true}

\newcommand{\fullonly}[1]{\ifthenelse{\boolean{display_full}}{#1}{}}
\newcommand{\lightonly}[1]{\ifthenelse{\boolean{display_light}}{#1}{}}

\newcommand{\Nat}{\mathbb{N}}

\newcommand{\Z}{\mathbb{Z}}
\newcommand{\Zrond}{\mathcal{Z}}

\newcommand{\C}{\EuScript{C}}
\newcommand{\X}{\mathcal{X}}

\newcommand{\finiteset}[2]{\{#1,\ldots,#2\}}

\newcommand{\moins}{\backslash}

\newcommand{\one}{\mathbb{1}}

\newcommand{\crochet}[1]{\left\llbracket#1\right\rrbracket}

\newcommand{\Int}{\mathcal{I}}
\newcommand{\act}[1]{\mathsf{act}(#1)}
\newcommand{\mulpp}{\mathsf{mul}_{+}}
\newcommand{\mulmm}{\mathsf{mul}_{-}}

\newcommand{\st}{\ |\ }

\newcommand{\Prog}{\EuScript{P}}
\newcommand{\ExaProg}{\EuScript{E}}
\newcommand{\LTS}{\EuScript{S}}
\newcommand{\CS}{\EuScript{P}}
\newcommand{\Aut}{\EuScript{A}}
\newcommand{\pow}[1]{{\mathbb P}(#1)}

\newcommand{\Var}{\mathcal{X}}
\newcommand{\Id}{\one}

\newcommand{\rel}[2]{\mathcal{R}\ifthenelse{\equal{#1}{}}{}{_{#1}}\ifthenelse{\equal{#2}{}}{}{({#2})}}
\newcommand{\post}{\mathsf{post}}
\newcommand{\poststar}[1]{\post^*\ifthenelse{\equal{#1}{}}{}{_{#1}}}

\newcommand{\dbrkts}[1]{\crochet{#1}}

\newcommand{\lst}[1]{\Lambda_{#1}}

\newcommand{\varren}{\kappa}
\newcommand{\valren}{\overrightarrow{\varren}}
\newcommand{\backvalren}{\overleftarrow{\varren}}

\newcommand{\onlyforpaperversion}[2]{%
  \ifthenelse{%
    \equal{#1}{\paperversion}}%
  {#2}{}%
}

\lstdefinelanguage{myalgo}{%
  morekeywords={if,then,else,repeat,while,for,to,forall,do,compute,call,return}
}

\newcommand{\paperversion}{full}

\begin{document}

\onlyforpaperversion{full}{\pagestyle{plain}}

\title{Accelerated Data-flow Analysis} 

\author{
  Jérôme Leroux \and
  Grégoire Sutre
}

\institute{
  LaBRI, Université de Bordeaux, CNRS\\
  Domaine Universitaire, 351, cours de la Libération, 33405 Talence, France\\
  \email{\{leroux, sutre\}@labri.fr}
}

\maketitle


\begin{abstract}
Acceleration in symbolic verification consists in computing the exact effect of some control-flow loops in order to speed up the iterative fix-point computation of reachable states.
Even if no termination guarantee is provided in theory, successful results were obtained in practice by different tools implementing this framework. In this paper, the acceleration framework is extended to data-flow analysis.  Compared to a classical widening/narrowing-based abstract interpretation, the loss of precision is controlled here by the choice of the abstract domain and does not depend on the way the abstract value is computed.
Our approach is geared towards precision, but we don't loose efficiency on the way.  Indeed,
we provide a cubic-time acceleration-based algorithm for solving interval constraints with full multiplication.
\end{abstract}

\section{Introduction}
\label{sec:introduction}

Model-checking safety properties on a given system usually reduces to the computation of a precise enough invariant of the system.  In traditional symbolic verification, the set of all reachable (concrete) configurations is computed iteratively from the initial states by a standard fix-point computation.  This reachability set is the most precise invariant, but quite often (in particular for software systems) a much coarser invariant is sufficient to prove correctness of the system.  Data-flow analysis, and in particular abstract interpretation~\cite{Cousot:1977:POPL}, provides a powerful framework to develop analysis for computing such approximate invariants.

\smallskip

A data-flow analysis of a program basically consists in the choice of a (potentially infinite) complete lattice of data properties for program variables together with transfer functions for program instructions.  The merge over all path (MOP) solution, which provides the most precise abstract invariant, is in general over-approximated by the minimum fix-point (MFP) solution, which is computable by Kleene fix-point iteration.  However the computation may diverge and \emph{widening/narrowing operators} are often used in order to enforce convergence at the expense of precision \cite{Cousot:1977:POPL, Cousot:1992:PLILP}.  While often providing very good results, the solution computed with widenings and narrowings may not be the MFP solution.  This may lead to abstract invariants that are too coarse to prove safety properties on the system under check.

\smallskip

Techniques to help convergence of Kleene fix-point iterations have also been investigated in symbolic verification of infinite-state systems.  In these works, the objective is to compute the (potentially infinite) reachability set for automata with variables ranging over unbounded data, such as counters, clocks, stacks or queues.  So-called \emph{acceleration} techniques (or \emph{meta-transitions}) have been developped~\cite{Boigelot:1994:CAV, Boigelot:1997:SAS, Comon:1998:CAV, Finkel:2003:IC, Finkel:2002:FSTTCS} to speed up the iterative computation of the reachability set.  Basically, acceleration consists in computing in one step the effect of iterating a given loop (of the control flow graph).  Accelerated symbolic model checkers such as \textsc{Lash}~\cite{LASH}, \textsc{TReX}~\cite{ABS-CAV01}, and \textsc{Fast}~\cite{Bardin:2003:CAV} successfully implement this approach.

\paragraph{Our contribution. }
In this paper, we extend acceleration techniques to data-flow analysis and we apply these ideas to interval analysis.  Acceleration techniques for (concrete) reachability set computations may be equivalently formalized ``semantically'' in terms of control-flow path languages~\cite{Leroux:2005:ATVA} or ``syntactically'' in terms of control-flow graph unfoldings~\cite{Bardin:2005:ATVA}.  We extend these concepts to the
MFP solution in a generic data-flow analysis framework, and we establish several links between the resulting notions.  It turns out that, for data-flow analysis, the resulting ``syntactic'' notion, based on graph \emph{flattenings}, is more general that the resulting ``semantic'' notion, based on restricted regular expressions.  We then propose a generic flattening-based semi-algorithm for computing the MFP solution.  This semi-algorithm may be viewed as a generic template for applying acceleration-based techniques to constraint solving.

\smallskip

We then show how to instantiate the generic flattening-based semi-algorithm in order to obtain an efficient constraint solver\footnote{By solver, we mean an algorithm computing the least solution of constraint systems.} for integers, for a rather large class of constraints using addition, (monotonic) multiplication, factorial, or any other \emph{bounded-increasing} function.  The intuition behind our algorithm is the following: we propagate constraints in a breadth-first manner as long as the least solution is not obtained, and variables involved in a ``useful'' propagation are stored in a graph-like structure.  As soon as a cycle appears in this graph, we compute the least solution of the set of constraints corresponding to this cycle.  It turns out that this acceleration-based algorithm always terminates in cubic-time.

\smallskip

As the main result of the paper, we then show how to compute in cubic-time the least solution for interval constraints with full addition and multiplication, and intersection with a constant.  The proof uses a least-solution preserving translation from interval constraints to the class of integer constraints introduced previously.

\paragraph{Related work. }
In \cite{Karr:1976:ACTA}, Karr presented a polynomial-time algorithm that computes the set of all affine relations that hold in a given control location of a (numerical) program.  Recently, the complexity of this algorithm was revisited in \cite{MullerOlm:2004:ICALP} and a fine upper-bound was presented.
For interval constraints with affine transfer functions, the exact least solution may be computed in cubic-time~\cite{Su:2004:TACAS}.  Strategy iteration was proposed in~\cite{CGGMP-CAV05} to speed up Kleene fix-point iteration with better precision than widenings and narrowings, and this approach has been developped in~\cite{GS-ESOP07} for interval constraint solving with full addition, multiplication and intersection.  Strategy iteration may be viewed as an instance of our generic flattening-based semi-algorithm.  The class of interval constraints that we consider in this paper contains the one in~\cite{Su:2004:TACAS} (which does not include interval multiplication) but it is more restrictive than the one in~\cite{GS-ESOP07}.  We are able to maintain the same cubic-time complexity as in~\cite{Su:2004:TACAS}, and it is still an open problem whether interval constraint solving can be performed in polynomial-time for the larger class considered in~\cite{GS-ESOP07}.


\paragraph{Outline. }
The paper is organized as follows.  Section~2 presents our acceleration-based approach to data-flow analysis.  We then focus on interval constraint-based data-flow analysis.  We present in section~3 a cubic-time algorithm for solving a large class of constraints over the integers, and we show in section~4 how to translate interval constraints (with multiplication) into the previous class of integer constraints, hence providing a cubic-time algorithm for interval constraints.  Section~5 presents some ideas for future work.
\onlyforpaperversion{proc}{
Please note that due to space constraints, most proofs are only sketched in this paper.  A long version of the paper with detailed proofs can be obtained from the authors.
}
\onlyforpaperversion{full}{
Please note that most proofs are only sketched in the paper, but detailed proofs are given in appendix.  This paper is the long version of our SAS 2007 paper.
}

%
%
%
%

\section{Acceleration in Data Flow Analysis}
\label{sec:dataflowanalysis}

This section is devoted to the notion of acceleration in the context of data-flow analysis.  Acceleration techniques for (concrete) reachability set computations~\cite{Boigelot:1994:CAV, Boigelot:1997:SAS, Comon:1998:CAV, Finkel:2003:IC, Finkel:2002:FSTTCS, Leroux:2005:ATVA, Bardin:2005:ATVA} may be equivalently formulated in terms of control-flow path languages or control-flow graph unfoldings.  We shall observe that this equivalence does not hold anymore when these notions are lifted to data-flow analysis.  All results in this section can easily be derived from the definitions, and they are thus presented without proofs.

\subsection{Lattices, words and graphs} 

We respectively denote by $\Nat$ and $\Z$ the usual sets of nonnegative integers  and integers.  For any set $S$, we write $\pow{S}$ for the set of subsets of $S$.  The \emph{identity} function over $S$ is written $\Id_{S}$, and shortly $\Id$ when the set $S$ is clear from the context.

\medskip

Recall that a \emph{complete lattice} is any partially ordered set $(A, \sqsubseteq)$ such that every subset $X \subseteq A$ has a \emph{least upper bound} $\bigsqcup X$ and a \emph{greatest lower bound} $\bigsqcap X$.  The \emph{supremum} $\bigsqcup A$ and the \emph{infimum} $\bigsqcap A$ are respectively denoted by $\top$ and $\bot$.  A function $f \in A \rightarrow A$ is \emph{monotonic}
if $f(x) \sqsubseteq f(y)$ for all $x \sqsubseteq y$ in $A$.
Recall that from Knaster-Tarski's Fix-point Theorem, any monotonic function $f \in A \rightarrow A$ has a \emph{least fix-point} given by $\bigsqcap \, \{a \in A \st f(a) \sqsubseteq a\}$.  For any monotonic function $f \in A \rightarrow A$, we denote by $f^*$ the monotonic function in $A \rightarrow A$ defined by $f^*(x) = \bigsqcap \, \{a \in A \st (x \sqcup f(a)) \sqsubseteq a\}$, in other words $f^*(x)$ is the least post-fix-point of $f$ greater than $x$.

\smallskip

For any complete lattice $(A, \sqsubseteq)$ and any set $S$, we also denote by $\sqsubseteq$ the partial order on $S \rightarrow A$ defined as the point-wise extension of $\sqsubseteq$, i.e. $f \sqsubseteq g$ iff $f(x) \sqsubseteq g(x)$ for all $x \in S$.  The partially ordered set $(S \rightarrow A, \sqsubseteq)$ is also a complete lattice, with lub $\bigsqcup$ and glb $\bigsqcap$ satisfying $(\bigsqcup F)(s) = \bigsqcup \, \{f(s) \st f \in F\}$ and $(\bigsqcap F)(s) = \bigsqcap \, \{f(s) \st f \in F\}$ for any subset $F \subseteq S \rightarrow A$.
Given any integer $n \geq 0$, we denote by $A^n$ the set of $n$-tuples over $A$.  We identify $A^n$ with the set $\{1, \ldots, n\} \rightarrow A$, and therefore $A^n$ equipped with the point-wise extension of $\sqsubseteq$ also forms a complete lattice.

\medskip

Let $\Sigma$ be an \emph{alphabet} (a finite set of \emph{letters}).  We write $\Sigma^*$ for the set of all (finite) \emph{words} $l_0 \cdots l_n$ over $\Sigma$, and $\varepsilon$ denotes the empty word.  Given any two words $x$ and $y$, we denote by $x \cdot y$ (shortly written $x y$) their \emph{concatenation}.   A subset of $\Sigma^*$ is called a \emph{language}.

\medskip

A (directed) \emph{graph} is any pair $G = (V, \rightarrow)$ where $V$ is a set of \emph{vertices} and $\rightarrow$ is a binary relation over $V$.  A pair $(v, v')$ in $\rightarrow$ is called an \emph{edge}.  A (finite) \emph{path} in $G$ is any (non-empty) sequence $v_0, \ldots, v_k$ of vertices, also written $v_0 \rightarrow v_1 \cdots v_{k-1} \rightarrow v_k$, such that $v_{i-1} \rightarrow v_i$ for all $1 \leq i \leq k$.  The nonnegative integer $k$ is called the \emph{length} of the path, and the vertices $v_0$ and $v_k$ are respectively called the \emph{source} and \emph{target} of the path.  A \emph{cycle} on a vertex $v$ is any path of non-zero length with source and target equal to $v$.  A cycle with no repeated vertices other than the source and the target is called \emph{elementary}.  We write $\stackrel{*}{\rightarrow}$ for the reflexive-transitive closure of $\rightarrow$.  A \emph{strongly connected component} (shortly \emph{SCC}) in $G$ is any equivalence class for the equivalence relation $\stackrel{*}{\leftrightarrow}$ on $V$ defined by: $v \stackrel{*}{\leftrightarrow} v'$ if $v \stackrel{*}{\rightarrow} v'$ and $v' \stackrel{*}{\rightarrow} v$.  We say that an SCC is \emph{cyclic} when it contains a unique elementary cycle up to cyclic permutation.

\subsection{Programs and data-flow solutions} 

For the rest of this section, we consider a complete lattice $(A, \sqsubseteq)$.  In our setting, a program will represent an instance (for some concrete program) of a data-flow analysis framework over $(A, \sqsubseteq)$.  To simplify the presentation, we will consider programs given as unstructured collections of commands (this is not restrictive as control-flow may be expressed through variables).

\medskip

Formally, assume a finite set $\Var$ of \emph{variables}.  A \emph{command} on $\Var$ is any tuple $\langle X_1, \ldots, X_n ; f ; X \rangle$, also written $X := f(X_1, \ldots, X_n)$, where $n \in \Nat$ is an \emph{arity}, $X_1, \ldots, X_n \in \Var$ are pairwise disjoint \emph{input variables}, $f \in A^n \rightarrow A$ is a monotonic \emph{transfer function}, and $X \in \Var$ is an \emph{output variable}.  Intuitively, a command $X := f(X_1, \ldots, X_n)$ assigns variable $X$ to $f(X_1, \ldots, X_n)$ and lets all other variables untouched.
A \emph{valuation} on $\Var$ is any function $\rho$ in $\Var \rightarrow A$.  The \emph{data-flow semantics} $\dbrkts{c}$ of any command $c = \langle X_1, \ldots, X_n ; f ; X \rangle$ on $\Var$ is the monotonic function in $(\Var \rightarrow A) \rightarrow (\Var \rightarrow A)$ defined by $\dbrkts{c}\!(\rho)(X) = f(\rho(X_1), \ldots, \rho(X_n))$ and $\dbrkts{c}\!(\rho)(Y) = \rho(Y)$ for all $Y \neq X$.

\medskip

A \emph{program} over $(A, \sqsubseteq)$ is any pair $\Prog = (\Var, C)$ where $\Var$ is a finite set of \emph{variables} and $C$ is a finite set of commands on $\Var$.

\begin{example} \label{ex:program}
  Consider the C-style source code given on the left-hand side below, that we want to analyse with the complete lattice $(\Int, \sqsubseteq)$ of intervals of $\Z$.  The corresponding program $\ExaProg$ is depicted graphically on the right-hand side below.
  \begin{center}
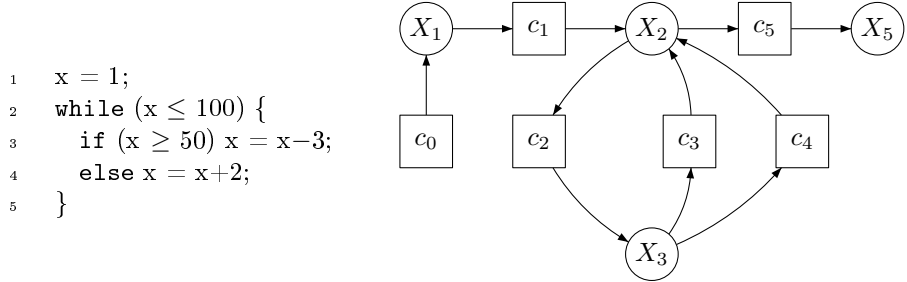
    \begin{minipage}{4.75cm}
      \begin{lstlisting}[language=C,frame=none]
  x = 1;
  while (x $\leq$ 100) {
    if (x $\geq$ 50) x = x-3;
    else x = x+2;
  }
      \end{lstlisting}
    \end{minipage}%
    \begin{minipage}{6.5cm}
      \begin{picture}(70,40)(-35,-20)
    \gasset{Nw=7,Nh=7}
    \node(X1)(-30,15){$X_1$}
    \node(X2)(0,15){$X_2$}
    \node(X3)(0,-15){$X_3$}
    \node(X4)(30,15){$X_5$}

    \node[Nmr=0.0](c0)(-30,0){$c_0$}
    \node[Nmr=0.0](c1)(-15,15){$c_1$}
    \node[Nmr=0.0](c2)(-15,0){$c_2$}
    \node[Nmr=0.0](c3)(5,0){$c_3$}
    \node[Nmr=0.0](c4)(20,0){$c_4$}
    \node[Nmr=0.0](c5)(15,15){$c_5$}

    \drawedge(c0,X1){}
    \drawedge(X1,c1){}
    \drawedge(c1,X2){}
    \drawedge[curvedepth=-2.0](X2,c2){}
    \drawedge[curvedepth=-2.0](c2,X3){}
    \drawedge[curvedepth=-2.0](X3,c3){}
    \drawedge[curvedepth=-2.0](X3,c4){}
    \drawedge[curvedepth=-2.0](c3,X2){}
    \drawedge[curvedepth=-2.0](c4,X2){}
    \drawedge(X2,c5){}
    \drawedge(c5,X4){}
      \end{picture}
    \end{minipage}%
  \end{center}
  Formally, the set of variables of $\ExaProg$ is $\{X_1, X_2, X_3, X_5\}$, representing the value of the variable $x$ at program points 1, 2, 3 and 5.  The set of commands of $\ExaProg$ is $\{c_0, c_1, c_2, c_3, c_4, c_5\}$, with:
  $$
  \begin{array}{rcl@{\qquad\quad}rcl}
    c_0 & : \ & X_1 := \top                           & c_3 & : \ & X_2 := (X_3 \,\sqcap\, [50, +\infty]) - \{3\} \\
    c_1 & : \ & X_2 := (\{0\} \,.\, X_1) + \{1\}      & c_4 & : \ & X_2 := (X_3 \,\sqcap\, ]-\infty, 49]) + \{2\} \\
    c_2 & : \ & X_3 := X_2 \,\sqcap\, ]-\infty, 100]  & c_5 & : \ & X_5 := X_2 \,\sqcap\, [101, +\infty[
  \end{array}
  $$
\end{example}

We will use language-theoretic terminology and notations for traces in a program.  A \emph{trace} in $\Prog$ is any word $c_1 \cdots c_k$ over $C$.  The empty word $\varepsilon$ denotes the empty trace and $C^*$ denotes the set of all traces in $\Prog$.  The data-flow semantics is extended to traces in the obvious way: $\dbrkts{\varepsilon} = \Id$ and $\dbrkts{c \cdot \sigma} = \dbrkts{\sigma} \circ \dbrkts{c}$.  Observe that $\dbrkts{\sigma \cdot \sigma'} = \dbrkts{\sigma'} \circ \dbrkts{\sigma}$ for every $\sigma, \sigma' \in C^*$.  We also extend the data-flow semantics to sets of traces by $\dbrkts{L} = \bigsqcup_{\sigma \in L} \dbrkts{\sigma}$ for every $L \subseteq C^*$.  Observe that $\dbrkts{L}$ is a monotonic function in $(\Var \rightarrow A) \rightarrow (\Var \rightarrow A)$, and moreover $\dbrkts{L_1 \cup L_2} = \dbrkts{L_1} \sqcup \dbrkts{L_2}$ for every $L_1, L_2 \subseteq C^*$.

\medskip

Given a program $\Prog = (\Var, C)$ over $(A, \sqsubseteq)$,
the \emph{minimum fix-point solution} (\emph{MFP-solution)} of $\Prog$, written $\lst{\Prog}$, is the valuation defined as follows:
\begin{eqnarray*}
  \lst{\Prog} & \ = \ & \bigsqcap \ \left\{ \rho \in \Var \rightarrow A \st \dbrkts{c}\!(\rho) \sqsubseteq \rho \mbox{ for all } c \in C \right\}
\end{eqnarray*}

\begin{example} \label{ex:program-solutions}
  The MFP-solution of the program $\ExaProg$ from Example~\ref{ex:program} is the valuation:
  $$
  \lst{\ExaProg} \ = \ \{X_1 \mapsto \top, \ X_2 \mapsto [1, 51], \ X_3 \mapsto [1, 51], \ X_5 \mapsto \bot\}
  $$
\end{example}

%
Recall that we denote by $\dbrkts{C}^*\!(\rho)$ the least post-fix-point of $\dbrkts{C}$ greater than $\rho$.  Therefore it follows from the definitions that $\lst{\Prog} = \dbrkts{C}^*\!(\bot)$.  In our framework, the \emph{merge over all paths solution} (\emph{MOP-solution)} may be defined as the valuation $\dbrkts{C^*}\!(\bot)$, and the following proposition recalls well-known links between the MOP-solution, the MFP-solution and the ascending Kleene chain.

\begin{proposition} \label{prop:mop-mfp-and-kleene}
  For any program $\Prog = (\Var, C)$ over a complete lattice $(A, \sqsubseteq)$, we have:
  $$
  \dbrkts{C^*}\!(\bot) \quad \sqsubseteq \quad \bigsqcup_{k \in \Nat} \ \dbrkts{C}^k\!(\bot) \quad \sqsubseteq \quad \dbrkts{C}^*\!(\bot) \quad = \quad \lst{\Prog}
  $$
\end{proposition}

\subsection{Accelerability and flattening} 

We now extend notions from accelerated symbolic verification to this data-flow analysis framework.  Acceleration in symbolic verification was first introduced semantically, in the form of \emph{meta-transitions}~\cite{Boigelot:1994:CAV, Boigelot:1997:SAS}, which basically simulate the effect of taking a given control-flow loop arbitrarily many times.  This leads us to the following proposition and definition.


\begin{proposition}
  Let $\Prog = (\Var, C)$ denote a program over $(A, \sqsubseteq)$.  For any languages $L_1, \ldots, L_k \subseteq C^*$, we have $(\dbrkts{L_k}^* \circ \cdots \circ \dbrkts{L_1}^*)(\bot) \ \sqsubseteq \ \lst{\Prog}$.
\end{proposition}

\begin{definition}
  A program $\Prog = (\Var, C)$ over a complete lattice $(A, \sqsubseteq)$ is called
    \emph{MFP-accelerable}
if
    $\lst{\Prog}   =   (\dbrkts{\sigma_k}^* \circ \cdots \circ \dbrkts{\sigma_1}^*)(\bot)$
  for some words $\sigma_1, \ldots, \sigma_k \in C^*$.
\end{definition}

The following proposition shows that any program $\Prog$ for which the ascending Kleene chain stabilizes after finitely many steps is MFP-accelerable.
\begin{proposition}
  Let $\Prog = (\Var, C)$ denote a program over $(A, \sqsubseteq)$.  If we have $\dbrkts{C}^k\!(\bot) = \lst{\Prog}$ for some $k \in \Nat$, then $\Prog$ is MFP-accelerable.
\end{proposition}

Acceleration in symbolic verification was later expressed syntactically, in terms of flat graph unfoldings.  When lifted to data-flow analysis, this leads to a more general concept than accelerability,
and we will show that these two notions coincide for ``concrete'' programs (as in symbolic verification).  We say that a program $\Prog$ is \emph{single-input} if the arity of every command in $\Prog$ is at most $1$.

\medskip

Given a program $\Prog = (\Var, C)$ over $(A, \sqsubseteq)$, an \emph{unfolding} of $\Prog$ is any pair $(\Prog', \varren)$ where $\Prog' = (\Var', C')$ is a program and $\varren \in \Var' \rightarrow \Var$ is a variable \emph{renaming}, and such that $\langle \varren(X'_1), \ldots, \varren(X'_n) ; f ; \varren(X') \rangle$ is a command in $C$ for every command $\langle X'_1, \ldots, X'_n ; f ; X' \rangle$ in $C'$.  The renaming $\varren$ induces a Galois surjection $(\Var' \rightarrow A, \sqsubseteq) \galois{\valren}{\backvalren} (\Var \rightarrow A, \sqsubseteq)$ where $\backvalren$ and $\valren$ are defined as expected by $\displaystyle \backvalren(\rho) = \rho \circ \kappa$ and $\displaystyle \valren(\rho')(X) = \bigsqcup_{\varren(X') = X} \rho'(X')$.

\medskip

We associate a bipartite graph to any program in a natural way: vertices are either variables or commands, and edges denote input and output variables of commands.  Formally, given a program $\Prog = (\Var, C)$, the \emph{program graph} of $\Prog$ is the labeled graph $G_{\Prog}$ where $\Var \cup C$ is the set of vertices and with edges $(c, X)$ and $(X_i, c)$ for every command $c = \langle X_1, \ldots, X_n ; f ; X \rangle$ in $C$ and $1 \leq i \leq n$.  We say that $\Prog$ is \emph{flat} if there is no SCC in $G_{\Prog}$ containing two distinct commands with the same output variable.  A \emph{flattening} of $\Prog$ is any unfolding $(\Prog', \varren)$ of $\Prog$ such that $\Prog'$ is flat.

\begin{example}
  A flattening of the program $\ExaProg$ from Example~\ref{ex:program} is given below.  Intuitively, this flattening represents a possible unrolling of the while-loop where the two branches of the inner conditional alternate.
  
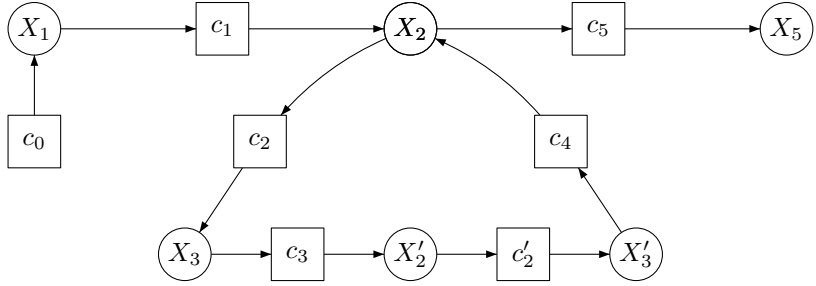
\begin{center}
    \begin{picture}(110,40)(-55,-20)
    \gasset{Nw=7,Nh=7}
    \node(X1)(-50,15){$X_1$}
    \node(X2)(0,15){$X_2$}
    \node(X'2)(0,-15){$X'_2$}
    \node(X3)(-30,-15){$X_3$}
    \node(X'3)(30,-15){$X'_3$}
    \node(X2)(0,15){$X_2$}
    \node(X4)(50,15){$X_5$}

    \node[Nmr=0.0](c0)(-50,0){$c_0$}
    \node[Nmr=0.0](c1)(-25,15){$c_1$}
    \node[Nmr=0.0](c2)(-20,0){$c_2$}
    \node[Nmr=0.0](c4)(20,0){$c_4$}
    \node[Nmr=0.0](c5)(25,15){$c_5$}
    \node[Nmr=0.0](c'2)(15,-15){$c'_2$}
    \node[Nmr=0.0](c3)(-15,-15){$c_3$}

    \drawedge(c0,X1){}
    \drawedge(X1,c1){}
    \drawedge(c1,X2){}
    \drawedge[curvedepth=-2.0](X2,c2){}
    \drawedge(c2,X3){}
    \drawedge(X3,c3){}
    \drawedge(c3,X'2){}
    \drawedge(X'2,c'2){}
    \drawedge(c'2,X'3){}
    \drawedge(X'3,c4){}
    \drawedge[curvedepth=-2.0](c4,X2){}
    \drawedge(X2,c5){}
    \drawedge(c5,X4){}
      \end{picture}
  \end{center}
\end{example}

\begin{lemma} \label{lemma:fix-point-under-approximation}
  Let $\Prog = (\Var, C)$ denote a program over $(A, \sqsubseteq)$.  For any unfolding $(\Prog', \varren)$ of $\Prog$, with $\Prog' = (\Var', C')$, we have $\valren \circ \dbrkts{C'}^* \circ \backvalren \ \sqsubseteq \ \dbrkts{C}^*$.
\end{lemma}

\begin{proposition}
  Let $\Prog = (\Var, C)$ denote a program over $(A, \sqsubseteq)$.  For any unfolding $(\Prog', \varren)$ of $\Prog$, we have $\valren(\lst{\Prog'}) \ \sqsubseteq \ \lst{\Prog}$.
\end{proposition}

%
\begin{definition}
  A program $\Prog = (\Var, C)$ over a complete lattice $(A, \sqsubseteq)$ is called
    \emph{MFP-flattable}
if
    $\lst{\Prog}   =   \valren(\lst{\Prog'})$
  for some flattening $(\Prog', \varren)$ of $\Prog$.
\end{definition}

Observe that any flat program is trivially MFP-flattable.
The following proposition establishes
links between accelerability and flattability.  As a corollary to the proposition, we obtain that
MFP-accelerability and MFP-flattability
are
equivalent for single-input
programs.

\begin{proposition}
  The following relationships hold for programs over $(A, \sqsubseteq)$:
  \begin{enumerate}
  \item[$i)$]   MFP-accelerability implies MFP-flattability.
  \item[$ii)$]  MFP-flattability implies MFP-accelerability for single-input programs.
  \end{enumerate}
\end{proposition}
\begin{proof}[Sketch]
  To prove $i)$,
  we use the fact that for every words $\sigma_1, \ldots, \sigma_k \in C^*$, there exists a finite-state automaton $\Aut$ without nested cycles recognizing $\sigma_1^* \cdots \sigma_k^*$.  The ``product'' of any program $\Prog$ with $\Aut$ yields a flattening that ``simulates'' the effect of $\sigma_1^* \cdots \sigma_k^*$ on $\Prog$.  To prove $ii)$,
  we observe that for any flat single-input program $\Prog$, each non-trivial SCC of $G_{\Prog}$ is cyclic.  We pick a ``cyclic'' trace (which is unique up to circular permutation) for each SCC, and we arrange these traces to prove that $\Prog$ is accelerable.  Backward preservation of accelerability under unfolding
  concludes the proof.
  \qed
\end{proof}

\begin{remark}
  For any labeled transition system $\LTS$ with a set $S$ of states, the forward collecting semantics of $\LTS$ may naturally be given as a single-input
  program $\Prog_{\LTS}$ over $(\pow{S}, \subseteq)$.  With respect to this translation (from $\LTS$ to $\Prog_{\LTS}$), the notion of flattability developped for accelerated symbolic verification of labeled transition systems coincide with the notions of MFP-accelerability and MFP-flattability defined above.
\end{remark}

Recall that our main goal is to compute (exact)
MFP-solutions
using accele\-ration-based techniques.  According to the previous propositions, flattening-based computation of the MFP-solution seems to be the most promising approach, and we will focus on this approach for the rest of the paper.

\subsection{Generic flattening-based constraint solving} 

It is well known that the MFP-solution of a program may also be expressed as the least solution of a constraint system, and we will use this formulation for the rest of the paper.  We will use some new terminology to reflect this new formulation, however notations and definitions will remain the same.  A command $\langle X_1, \ldots, X_n ; f ; X \rangle$ will now be called a \emph{constraint}, and will also be written $X \sqsupseteq f(X_1, \ldots, X_n)$.  A program over $(A, \sqsubseteq)$ will now be called a \emph{constraint system} over $(A, \sqsubseteq)$, and the MFP-solution will be called the \emph{least solution}.  Among all acceleration-based notions defined previously, we will only consider MFP-flattability for constraint systems, and hence we will shortly write \emph{flattable} instead of MFP-flattable.

\medskip

Given a constraint system $\CS = (\Var, C)$ over $(A, \sqsubseteq)$, any valuation $\rho \in \Var \rightarrow A$ such that $\rho \sqsubseteq \dbrkts{C}\!(\rho)$ (resp. $\rho \sqsupseteq \dbrkts{C}\!(\rho)$) is called a \emph{pre-solution} (resp. a \emph{post-solution}).  A post-solution is also shortly called a \emph{solution}.  Observe that the least solution $\lst{\CS}$ is the greatest lower bound of all solutions of $C$.

\medskip

We now present a generic flattening-based semi-algorithm for constraint solving.  Intuitively, this semi-algorithm performs a propagation of constraints starting from the valuation $\bot$, but at each step we extract a flat ``subset'' of constraints (possibly by duplicating some variables) and we update the current valuation with the least solution of this flat ``subset'' of constraints.

\medskip
\begin{lstlisting}[emph={Solve}]
  Solve($\CS = (\Var, C)$ : a constraint system)
  $\rho \leftarrow \bot$
  while $\dbrkts{C}\!(\rho) \not\sqsubseteq \rho$
       construct a flattening $(\Prog', \varren)$ of $\Prog$, where $\Prog' = (\Var', C')$
       $\rho'  \leftarrow \rho \circ \varren$
       $\rho'' \leftarrow \dbrkts{C'}^*\!(\rho')$ $\hspace{27mm}$ { $\valren(\rho'') \sqsubseteq \dbrkts{C}^*\!(\rho)$ from Lemma $\mbox{\ref{lemma:fix-point-under-approximation}}$ }
       $\rho   \leftarrow \rho \sqcup \valren(\rho'')$
  return $\rho$
\end{lstlisting}
\medskip

The \textsf{Solve} semi-algorithm may be viewed as a generic template for applying acceleration-based techniques to constraint solving.  The two main challenges are (1) the construction of a suitable flattening at line~4, and (2) the computation of the least solution for flat constraint systems (line~6).  However, assuming that all involved operations are effective, this semi-algorithm is \emph{correct} (i.e. if it terminates then the returned valuation is the least solution of input constraint system), and it is \emph{complete} for flattable constraint systems (i.e. the input constraint system is flattable if and only if there exists choices of flattenings at line~4 such that the while-loop terminates).
We will show in the sequel how to instantiate the \textsf{Solve} semi-algorithm in order to obtain an efficient constraint solver for integers and intervals.





\section{Integer Constraints}
Following \cite{Su:2004:TACAS, GS-ESOP07}, we first investigate integer constraint solving in order to derive in the next section an interval solver. This approach is motivated by the encoding of an interval by two integers.

\medskip

The \emph{complete lattice of integers} $\Zrond=\Z\cup\{-\infty,+\infty\}$ is equipped with the natural order:
$$-\infty< \cdots < -2< -1< 0<1<2<\cdots < +\infty$$
Observe that the least upper bound $x\vee y$ and the greatest lower bound $x\wedge y$ respectively correspond to the maximum and the minimum. Addition and multiplication functions are extended from $\Z$ to $\Zrond$ as in \cite{GS-ESOP07}:
$$\begin{array}{@{}cccccp{0.2cm}cccccp{0.2cm}l@{}}
  x.0         &=& 0.x         &=& 0       &
  &
  x+(-\infty) &=& (-\infty)+x &=& -\infty &   
  &
  \text{for all }x\\

  x.(+\infty) &=& (+\infty).x &=& +\infty &   
  &
  x.(-\infty) &=& (-\infty).x &=& -\infty &
  &
  \text{for all }x>0\\
  
  x.(+\infty) &=& (+\infty).x &=& -\infty &   
  &
  x.(-\infty) &=& (-\infty).x &=& +\infty &
  &
  \text{for all }x<0\\ 

  x+(+\infty) &=& (+\infty)+x &=& +\infty &   
  &
              & &             & &         &
  &
  \text{for all }x>-\infty\\
  
\end{array}$$


A constraint system $\CS=(\X,C)$ is said \emph{cyclic} if the set of constraints $C$ is contained in a cyclic SCC.  An example is given below.

\medskip

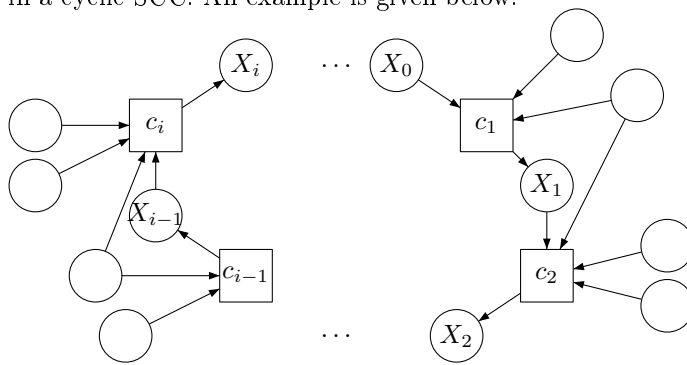
\begin{center}
\begin{picture}(99,40)(0,-50)
\gasset{Nw=7,Nh=7}
\node[NLangle=0.0](n0)(56.0,-12.0){$X_0$}

\node[NLangle=0.0,Nmr=0.0](n1)(68.0,-20.0){$c_1$}

\node[NLangle=0.0](n2)(76.0,-28.0){$X_1$}

\node[NLangle=0.0,Nmr=0.0](n3)(76.0,-40.0){$c_2$}

\node(n4)(80.0,-8.0){}

\node(n5)(88.0,-16.0){}

\node(n6)(92.0,-36.0){}

\node(n7)(92.0,-44.0){}

\node[NLangle=0.0](n8)(36.0,-12.0){$X_{i}$}

\node[NLangle=0.0](n9)(24.0,-32.0){$X_{i-1}$}

\node(n10)(16.0,-40.0){}

\node(n11)(20.0,-48.0){}

\node(n12)(8.0,-20.0){}

\node(n13)(8.0,-28.0){}

\drawedge(n0,n1){}

\drawedge(n1,n2){}

\drawedge(n2,n3){}

\node[NLangle=0.0,Nmr=0.0](n14)(24.0,-20.0){$c_{i}$}

\node[NLangle=0.0,Nmr=0.0](n15)(36.0,-40.0){$c_{i-1}$}

\drawedge(n14,n8){}

\drawedge(n15,n9){}

\drawedge(n9,n14){}

\drawedge(n12,n14){}

\drawedge(n13,n14){}

\drawedge(n10,n15){}

\drawedge(n11,n15){}

\drawedge(n4,n1){}

\drawedge(n5,n1){}

\drawedge(n6,n3){}

\drawedge(n7,n3){}

\node[NLangle=0.0](n16)(64.0,-48.0){$X_2$}

\drawedge(n3,n16){}

\node[Nframe=n,NLangle=0.0](n17)(48.0,-48.0){$\ldots$}

\node[Nframe=n,NLangle=0.0](n18)(48.0,-12.0){$\ldots$}

\drawedge(n5,n3){}

\drawedge(n10,n14){}

\end{picture}
\end{center}
Observe that a cyclic constraint system is flat. A \emph{cyclic flattening} $(\CS',\kappa)$ where $\CS'=(\X',C')$ can be naturally \emph{associated} to any cycle $X_0\rightarrow{c_1}\rightarrow X_1\cdots \rightarrow c_n\rightarrow X_n=X_0$ of a constraint system $\CS$, by considering the set $\X'$ of variables obtained from $\X$ by adding $n$ new copies $Z_1,\ldots, Z_n$ of $X_1,\ldots,X_n$ with the corresponding renaming $\kappa$ that extends the identity function over $\X$ by $\kappa(Z_i)=X_i$, and by considering the set of constraints $C'=\{c_1',\ldots,c_n'\}$ where $c_i'$ is obtained from $c_i$ by renaming the output variable $X_i$ by $Z_i$ and by renaming the input variable $X_{i-1}$ by $Z_{i-1}$ where $Z_0=Z_n$.\\

In section \ref{sub:bics}, we introduce an instance of the generic \textsf{Solve} semi-algorithm that solves constraint systems that satisfy a property called \emph{bounded-increasing}. This class of constraint systems is extended in section \ref{sub:ics} with test constraints allowing a natural translation of \emph{interval} constraint systems to contraint systems in this class.

\subsection{Bounded-increasing constraint systems}\label{sub:bics}
A monotonic function $f\in \Zrond^k\rightarrow\Zrond$ is said \emph{bounded-increasing} if for any $x_1<x_2$ such that $f(\bot)<f(x_1)$ and $f(x_2)<f(\top)$ we have $f(x_1)<f(x_2)$. Intuitively $f$ is increasing over the domain of $x\in\Zrond^k$ such that $f(x)\not\in\{f(\bot),f(\top)\}$.

\begin{example}
  The guarded identity $x\mapsto x\wedge b$ where $b\in\Zrond$, the addition $(x,y)\mapsto x+y$, the two multiplication functions $\mulpp$ and $\mulmm$ defined below, the power by two $x\mapsto 2^{x\vee 0}$, the factorial $x\mapsto !(x\vee 1)$ are bounded-increasing. However the minimum and the maximum functions are not bounded-increasing.
  $$
  \mulpp(x,y)=
  \begin{cases}
    x.y & \text{if } x,y\geq 0 \\
    0 & \text{otherwise}\\
  \end{cases}
  \;\;\;\;\;\;
  \mulmm(x,y)=
  \begin{cases}
    -x.y & \text{if } x,y<0 \\
    0 & \text{otherwise}\\
  \end{cases}
  $$
\end{example}

A \emph{bounded-increasing constraint} is a constraint of the form  $X\geq f(X_1,\ldots,X_k)$ where $f$ is a bounded-increasing function. Such a constraint is said \emph{upper-saturated} (resp. \emph{lower-saturated}) by a valuation $\rho$ if $\rho(X)\geq f(\top)$ (resp. $\rho(X)\leq f(\bot)$). Given a constraint system $\CS=(\X,C)$ and a bounded-increasing constraint $c\in C$ upper-saturated by a valuation $\rho_0$, observe that $\crochet{C}^*\!(\rho_0)=\crochet{C'}^*\!(\rho_0)$ where $C'=C\moins\{c\}$. Intuitively, an upper-saturated constraint for $\rho_0$ can be safely removed from a constraint system without modifying the least solution greater than $\rho_0$. The following lemma will be useful to obtain upper-saturated constraints.
\begin{lemma}\label{lem:sunny}
  Let $\CS$ be a cyclic bounded-increasing constraint system. If $\rho_0$ is a pre-solution of $\CS$ that does not lower-saturate any constraint, then either $\rho_0$ is a solution or $\crochet{C}^*\!(\rho_0)$ upper-saturates a constraint.
\end{lemma}
\begin{proof}
  \emph{(Sketch).}
  Let $X_0\rightarrow c_1\rightarrow X_1\rightarrow \cdots \rightarrow c_n \rightarrow X_n=X_0$ be the unique (up to a cyclic permutation) cycle in the graph associated to $\CS$. 
  Consider a pre-solution $\rho_0$ of $\CS$ that is not a solution. Let us denote by $(\rho_i)_{i\geq 0}$ the sequence of valuations defined inductivelly by $\rho_{i+1}=\rho_i\vee \crochet{C}\!(\rho_i)$.
  There are two cases:
  \begin{itemize}
  \item either there exists $i\geq 0$ such that $\rho_i$ upper-saturates a constraint $c_j$. Since $\rho_i\leq \crochet{C}^*\!(\rho_0)$, we deduce that $\crochet{C}^*\!(\rho_0)$ upper-saturates $c_j$.
  \item or $c_1,\ldots,c_n$ are not upper-saturated by any of the $\rho_i$. As these constraints are bounded-increasing, the sequence $(\rho_i)_{i\geq 0}$ is strictly increasing. Thus  $(\bigvee_{i\geq 0}\rho_i)(X_j)=+\infty$ for any $1\leq j\leq n$. Since  $\bigvee_{i\geq 0}\rho_i\leq \crochet{C}^*\!(\rho_0)$, we deduce that $\crochet{C}^*\!(\rho_0)$ upper-saturates $c_1,\ldots,c_n$.
  \end{itemize}
  In both cases, $\crochet{C}^*\!(\rho_0)$ upper-saturates at least one constraint.
  \qed
\end{proof}

\medskip
\begin{lstlisting}[emph={CyclicSolve}]
  CyclicSolve($\CS = (\Var, C)$ : a cyclic bounded-increasing constraint system,
              $\rho_0$ : a valuation)
  let $X_0\rightarrow c_1\rightarrow X_1\cdots \rightarrow c_n\rightarrow X_n=X_0$ be the ``unique'' elementary cycle
  $\rho \leftarrow \rho_0$
  for $i=1$ to $n$ do
       $\rho\leftarrow \rho\vee\crochet{c_i}\!(\rho)$
  for $i=1$ to $n$ do
       $\rho\leftarrow \rho\vee\crochet{c_i}\!(\rho)$
  if $\rho\geq \crochet{C}\!(\rho)$
       return $\rho$
  for $i=1$ to $n$ do
       $\rho(X_i)\leftarrow +\infty$
  for $i=1$ to $n$ do
       $\rho\leftarrow \rho\wedge\crochet{c_i}\!(\rho)$
  for $i=1$ to $n$ do
       $\rho\leftarrow \rho\wedge\crochet{c_i}\!(\rho)$
  return $\rho$
\end{lstlisting}
\medskip

\begin{proposition}
 The algorithm \textsf{CyclicSolve} returns $\crochet{C}^*\!(\rho_0)$ for any cyclic constraint system $\CS$ and for any valuation $\rho_0$.
\end{proposition}
\begin{proof}
  \emph{(Sketch).} The first two loops (lines 5--8) propagate the valuation $\rho_0$ along the cycle two times.  If the resulting valuation is not a solution at this point, then it is a pre-solution and no constraint is lower-saturated.  From Lemma~\ref{lem:sunny}, we get that $\crochet{C}^*\!(\rho_0)$ upper-saturates some constraint.  Observe that the valuation $\rho$ after the third loop (lines 11--12) satisfies $\crochet{C}^*\!(\rho_0) \sqsubseteq \rho$. The descending iteration of the last two loops yields (at line 17) $\crochet{C}^*\!(\rho_0)$.
  \qed
\end{proof}

We may now present our cubic time algorithm for solving bounded-increasing constraint systems. The main loop of this algorithm first performs $|C|+1$ rounds of Round Robin iterations and keeps track for each variable of the last constraint that updated its value. This information is stored in a partial function $\lambda$ from $\X$ to $C$. The second part of the main loop checks whether there exists a cycle in the subgraph induced by $\lambda$, and if so it selects such a cycle and calls the \textsf{CylicSolve} algorithm on it.

\medskip\medskip
\begin{lstlisting}[emph={SolveBI}]
  SolveBI($\CS=(\X,C)$ : a bounded-increasing constraint system,
               $\rho_0$ : an initial valuation)
  $\rho\leftarrow \rho_0 \vee \crochet{C}\!(\rho_0)$
  while $\dbrkts{C}\!(\rho) \not\sqsubseteq \rho$ 
       $\lambda\leftarrow\emptyset$                                $\{$ $\lambda$ is a partial function from $\X$ to $C$ $\}$
       repeat $|C|+1$ times
            for each $c\in C$
                 if $\rho\not\geq \crochet{c}\!(\rho)$
                      $\rho\leftarrow \rho\vee \crochet{c}\!(\rho)$
                      $\lambda(X)\leftarrow c$, where $X$ is the input variable of $c$
        if there exists an elementary cycle $X_0\rightarrow \lambda(X_1)\rightarrow X_1\cdots \lambda(X_n)\rightarrow X_0$
             construct the corresponding cyclic flattening $(\CS',\kappa)$
             $\rho'  \leftarrow \rho \circ \varren$
             $\rho'' \leftarrow \textsf{CyclicSolve}(\CS',\rho')$ 
             $\rho   \leftarrow \rho \vee \valren(\rho'')$
  return $\rho$
\end{lstlisting}
\medskip\medskip

Note that the \textsf{SolveBI} algorithm is an instance of the \textsf{Solve} semi-algorithm where flattenings are deduced from cycles induced by the partial function $\lambda$. 
The following proposition \ref{thm:flathalf} shows that this algorithm terminates.
\begin{proposition}\label{thm:flathalf}
  The algorithm \textsf{SolveBI} returns the least solution $\crochet{C}^*\!(\rho_0)$ of a bounded-increasing constraint system $\CS$ greater than a valuation $\rho_0$. Moreover, the number of times the while loop is executed is bounded by one plus the number of constraints that are upper-saturated for $\crochet{C}^*\!(\rho_0)$ but not for $\rho_0$.
\end{proposition}
\begin{proof}
  \emph{(Sketch).}
  Observe that initially $\rho=\rho_0\vee\crochet{C}\!(\rho_0)$. Thus, if during the execution of the algorithm $\rho(X)$ is updates by a constraint $c$ then necessary $c$ is not lower-saturated. That means if $\lambda(X)$ is defined then $c=\lambda(X)$ is not lower-saturated. 
  
  Let $\rho_0$ and $\rho_1$ be the values of $\rho$ respectively before and after the execution of the first two nested loops (line 5-9) and let $\rho_2$ be the value of $\rho$ after the execution of line~14.

  Observe that if there does not exist an elementary cycle satisfying the condition given in line 11, the graph associated to $\CS$ restricted to the edges $(X,c)$ if $c=\lambda(X)$ and the edges $(X_i,c)$ if $X_i$ is an input variable of $c$ is acyclic. This graph induces a natural partial order over the constraints $c$ of the form $c=\lambda(X)$. An enumeration $c_1,\ldots,c_m$ of this constraints compatible with the partial order provides the relation $\rho_1\leq \crochet{c_1\ldots c_m}\!(\rho_0)$. Since the loop 6-9 is executed at least $m+1$ times, we deduce that $\rho_1$ is a solution of $\C$.

  Lemma \ref{lem:sunny} shows that if $\rho_1$ is not a solution of $\CS$ then at least one constraint is upper-saturated for $\rho_2$ but not for $\rho_0$. We deduce that the number of times the while loop is executed is bounded by one plus the number of constraints that are upper-saturated for $\crochet{C}^*\!(\rho_0)$ but not for $\rho_0$.
  \qed
\end{proof}

\subsection{Integer constraint systems}\label{sub:ics}
A \emph{test function} is a function $\theta_{>b}$ or $\theta_{\geq b}$ with $b\in\Zrond$ of the following form:
$$\theta_{\geq b}(x,y)=
  \begin{cases}
    y & \text{if } x\geq b \\
    -\infty & \text{otherwise}\\
  \end{cases}
  \;\;\;\;\;\;
  \theta_{>b}(x,y)=
  \begin{cases}
    y & \text{if } x>b \\
    -\infty & \text{otherwise}\\
   \end{cases}
$$
A \emph{test constraint} is a constraint of the form $X\geq \theta_{\sim b}(X_1,X_2)$ where $\theta_{\sim b}$ is a test function. Such a constraint $c$ is said \emph{active} for a valuation $\rho$ if $\rho(X_1)\sim b$. Given a valuation $\rho$ such that $c$ is active, observe that $\crochet{c}\!(\rho)$ and $\crochet{c'}\!(\rho)$ are equal where $c'$ is the bounded-increasing constraint $X\geq X_2$. This constraint $c'$ is called the \emph{active form} of $c$ and denoted by $\act{c}$.

\medskip

In the sequel, an \emph{integer constraint} either refers to a bounded-increasing constraint or a test-constraint.

\medskip
\begin{lstlisting}[emph={SolveInteger}]
  SolveInteger($\CS=(\X,C)$ : an integer constraint system)
  $\rho\leftarrow \bot$
  $C_t\leftarrow$ set of test constraints in $C$
  $C'\leftarrow$ set of bounded-increasing constraints in $C$ 
  while $\dbrkts{C}\!(\rho) \not\sqsubseteq \rho$ 
       $\rho\leftarrow\textsf{SolveBI}((\X,C'),\rho)$
       for each $c\in C_t$
            if $c$ is active for $\rho$
                 $C_t\leftarrow C_t\moins\{c\}$
                 $C'\leftarrow C'\cup\{\act{c}\}$
  return $\rho$
\end{lstlisting}
\medskip

\begin{theorem}
  The algorithm \textsf{SolveInteger} computes the least solution of an integer constraint system $\CS=(\X,C)$ by performing $O((|\X|+|C|)^3)$ integer comparisons and image computation by some bounded-increasing functions. 
\end{theorem}
\begin{proof}
  Let us denote by $n_t$ be the number of test constraints in $C$. Observe that if during the execution of the while loop, no test constraints becomes active (line 7-10) then $\rho$ is a solution of $\CS$ and the algorithm terminates. Thus this loop is executed at most $1+n_t$ times. Let us denote by $m_1,\ldots,m_k$ the integers such that $m_i$ is equal to the number of times the while loop of \textsf{SolveBI} is executed. Since after the execution there is $m_i-1$ constraints that becomes upper-saturated, we deduce that $\sum_{i=1}^k(m_i-1)\leq n$ and in particular $\sum_{i=1}^km_i\leq n+k\leq 2.|C|$. Thus the algorithm \textsf{SolveInteger} computes the least solution of an integer constraint system $\CS=(\X,C)$ by performing $O((|\X|+|C|)^3)$ integer comparisons and image computation by some bounded-increasing functions. 
  \qed 
\end{proof}

\begin{remark}
  We deduce that any integer constraint system is MFP-flattable.
\end{remark}




\section{Interval Constraints}\label{sec:interval}
In this section, we provide a cubic time constraint solver for intervals. Our solver is based on the usual \cite{Su:2004:TACAS, GS-ESOP07} encoding of intervals by two integers in $\Zrond$.  The main challenge is the translation of an interval constraint system with full multiplication into an integer constraint system.

\medskip

An \emph{interval} $I$ is subset of $\Z$ of the form $\{x\in\Z;\;a\leq x\leq b\}$ where $a,b\in\Zrond$. We denote by $\Int$ the complete lattice of intervals partially ordered with the inclusion relation $\sqsubseteq$.
The \emph{inverse} $-I$ of an interval $I$, the \emph{sum} and the \emph{multiplication} of two intervals $I_1$ and $I_2$ are defined as follows:
$$
-I=\{-x;\; x\in I\}
\;\;\;\;\;\;\;\;
\begin{array}{@{}rclcl@{}}
I_1&+&I_2 &=& \{x_1+x_2;\;(x_1,x_2)\in I_1\times I_2\}\\
I_1&.&I_2 &=& \bigsqcup\{x_1.x_2;\;(x_1,x_2)\in I_1\times I_2\}\\
\end{array}
$$
We consider interval constraints of the following forms where $I\in \Int$:
$$
X\sqsupseteq - X_1
\;\;\;\;\; 
X\sqsupseteq I
\;\;\;\;\;
X\sqsupseteq X_1 + X_2 
\;\;\;\;\; 
X\sqsupseteq X_1 \sqcap I
\;\;\;\;\;
X\sqsupseteq X_1 . X_2
$$
Observe that we allow arbitrary multiplication between intervals, but we restrict intersection to intervals with a constant interval.


\medskip

We say that an interval constraint system $\CS=(\X,C)$ has the positive-multiplication property if for any constraint $c\in C$ of the form $X\sqsupseteq X_1.X_2$, the intervals $\lst{\CS}(X_1)$ and $\lst{\CS}(X_2)$ are included in $\Nat$. Given an interval constraint system $\CS=(\X,C)$ we can effectively compute an interval constraint system $\CS'=(\X',C')$ satisfying this property and such that $\X\subseteq \X'$ and $\lst{\CS}(X)=\lst{\CS'}(X)$ for any $X\in\X$. This constraint system $\CS'$ is obtained from $\CS$ by replacing the constraints $X\sqsupseteq X_1.X_2$ by the following constraints:
\begin{align*}
&X\sqsupseteq X_{1,u}.X_{2,u} &&   X_{1,u}\sqsupseteq X_1\sqcap \Nat\\
&X\sqsupseteq X_{1,l}.X_{2,l} &&   X_{2,u}\sqsupseteq X_2\sqcap \Nat\\
&X\sqsupseteq -X_{1,u}.X_{2,l} &&   X_{1,l}\sqsupseteq (-X_1)\sqcap \Nat\\
&X\sqsupseteq -X_{1,l}.X_{2,u} &&   X_{2,l}\sqsupseteq (-X_2)\sqcap \Nat
\end{align*}
Intuitively $X_{1,u}$ and $X_{2,u}$ corresponds to the positive parts of $X_1$ and $X_2$, while $X_{1,l}$ and $X_{2,l}$ corresponds to the negative parts.

\medskip

Let us provide our construction for translating an interval constraint system $\CS=(\X,C)$ having the positive multiplication property into an integer constraint system $\CS'=(\X',C')$. Since an interval $I$ can be naturally encoded by two integers $I^-,I^+\in\Zrond$ defined as the least upper bound of respectively $-I$ and $I$, we naturally assume that $\X'$ contains two integer variable $X^-$ and $X^+$ encoding each interval variable $X\in\X$. In order to extract from the least solution of $\CS'$ the least solution of $\CS$, we are looking for an integer constraint system $\CS'$ satisfying $(\lst{\CS}(X))^-=\lst{\CS'}(X^-)$ and $(\lst{\CS}(X))^+=\lst{\CS'}(X^+)$ for any $X\in\X$.

\medskip

As expected, a constraint $X\sqsupseteq -X_1$ is converted into $X^+\geq X_1^-$ and $X^-\geq X_1^+$, a constraint $X\sqsupseteq I$ into $X^+\geq I^+$ and $X^-\geq I^-$, and a constraint $X\sqsupseteq X_1+X_2$ into $X^-\geq X_1^- + X_2^-$ and $X^- \geq X_1^- + X_2^-$.
However, a constraint $X\sqsupseteq X_1\sqcap I$ cannot be simply translated into $X^-\geq X_1^-\wedge I^-$ and $X^+\geq X_1^+\wedge I^+$. In fact, these constraints may introduce imprecision when $\lst{\CS}(X)\cap I=\emptyset$. We use test functions to overcome this problem. Such a constraint is translated into the following integer constraints:
\begin{align*} 
&X^- \;\geq\; \theta_{\geq -I^+}(X_1^-,\theta_{\geq -I^-}(X_1^+,X_1^-\wedge I^-))\\
&X^+ \;\geq\; \theta_{\geq -I^-}(X_1^+,\theta_{\geq -I^+}(X_1^-,X_1^+\wedge I^+))
\end{align*}

For the same reason, the constraint $X\sqsupseteq X_1.X_2$ cannot be simply converted into $X^-\geq \mulmm(X_1^-,X_2^-)$ and $X^+\geq \mulpp(X_1^+,X_2^+)$. Instead, we consider the following constraints:
\begin{align*}
& X^- \;\geq\; \theta_{>-\infty}(X_1^-,\theta_{>-\infty}(X_1^+, \theta_{>-\infty}(X_2^-,\theta_{>-\infty}(X_2^+,\mulmm(X_1^-,X_2^-)))))\\
& X^+ \;\geq\; \theta_{>-\infty}(X_1^+,\theta_{>-\infty}(X_1^-, \theta_{>-\infty}(X_2^+,\theta_{>-\infty}(X_2^-,\mulpp(X_1^+,X_2^+)))))
\end{align*}
Observe in fact that $X^-\geq \mulmm(X_1^-,X_2^-)$ and $X^+\geq \mulpp(X_1^+,X_2^+)$ are precise constraint when the intervals $I_1=\lst{\CS}(X_1)$ and $I_2=\lst{\CS}(X_2)$ are non empty. Since, if this condition does not hold then $I_1.I_2=\emptyset$, the previous encoding consider this case by testing if the values of $X_1^-$, $X_1^+$, $X_2^-$, $X_2^+$ are strictly greater than $-\infty$.

\medskip

Now, observe that the integer constraint system $\CS'$ satisfies the equalities $(\lst{\CS}(X))^+=\lst{\CS'}(X^+)$ and $(\lst{\CS}(X))^-=\lst{\CS'}(X^-)$ for any $X\in\X$. Thus, we have proved the following theorem.
\begin{theorem}
   The least solution of an interval constraint system $\CS=(\X,C)$ with full multiplication can by computed in time $O((|\X|+|C|)^3)$ with integer manipulations performed in $O(1)$.
\end{theorem}
                
\begin{remark}
  We deduce that any interval constraint system is MFP-flattable. 
\end{remark}



\section{Conclusion and Future Work}
In this paper we have extended the acceleration framework from symbolic verification to the computation of MFP-solutions in data-flow analysis. Our approach leads to an efficient cubic-time algorithm for solving interval constraints with full addition and multiplication, and intersection with a constant.

\medskip

As future work, it would be interesting to combine this result with strategy iteration techniques considered in~\cite{GS-ESOP07} in order to obtain a polynomial time algorithm for the extension with full intersection. 
We also intend to investigate the application of the acceleration framework to other abstract domains.


\bibliographystyle{alpha}
\bibliography{this-biblio}


\onlyforpaperversion{full}{
\clearpage
\appendix

\newpage

\section{Proof of Lemma~3.2}
\gdef\thesection{3}
\setcounter{theorem}{1}

\begin{lemma}\label{lem:sunny}
  Let $\CS$ be a cyclic bounded-increasing constraint system. If $\rho_0$ is a pre-solution of $\CS$ that does not lower-saturate any constraint, then either $\rho_0$ is a solution or $\crochet{C}^*(\rho_0)$ upper-saturates a constraint.
\end{lemma}
\begin{proof}
  Let $X_0\rightarrow c_1\rightarrow X_1\rightarrow \cdots \rightarrow c_n \rightarrow X_n=X_0$ be the unique (up to a cyclic permutation) cycle in the graph associated to $\CS$.
  
  Let us prove that for any pre-solution $\rho$ that is not a solution and that does not lower-saturate any constraint, there exists a constraint $c\in C$ such that $\rho'=\crochet{c}(\rho)$ is a pre-solution satisfying $\rho'>\rho$ that either upper-saturates a constraint or that is not a solution. Since $\rho$ is not a solution, there exists a constraint $c_{i-1}$ such that the valuation $\rho'=\crochet{c_{i-1}}(\rho)$ satisfies $\rho'\not\leq \rho$. As $c_i$ only modifies the value of $X_{i-1}$, we get $\rho'(X_{i-1})>\rho(X_{i-1})$. Observe that if $\rho'$ upper-saturates $c_{i-1}$ we are done. Let us assume that $\rho'$ does not upper-saturate $c_{i-1}$. Let us show that $\rho'$ is not a solution of $C$. As $\rho$ is a pre-solution and $c_i$ is the unique variable that modifies $X_i$, we have $\rho(X_{i})\leq \crochet{c_{i}}(\rho)(X_{i})$. Since $\rho(X_{i-1})<\rho'(X_{i-1})$ and $c_i$ is neither upper-saturate nor lower-saturate for $\rho'$ and $\rho$ we get $\crochet{c_{i}}(\rho)(X_{i})<\crochet{c_{i}}(\rho')(X_{i})$ from $\rho<\rho'$. The relations $\rho(X_{i})=\rho'(X_{i})$, $\rho(X_{i})\leq \crochet{c_{i}}(\rho)(X_{i})$ and $\crochet{c_{i}}(\rho)(X_{i})<\crochet{c_{i}}(\rho')(X_{i})$ provide  the relation $\crochet{c_{i}}(\rho')(X_{i})>\rho'(X_{i})$. Thus $\rho'$ is not a solution.
  
  Assume by contradiction that $\crochet{C}^*(\rho_0)$ does not upper-saturate a constraint. Since $\rho_0$ is a pre-solution that is not a solution and such that any constraint $c\in C$ is not lower-saturated, from the previous paragraph, we get an infinite sequence $\rho_0<\ldots<\rho_k<\ldots$ of valuations satisfying $\rho_k\leq\crochet{C}^*(\rho_0)$. We deduce that there exists a variable $X_i$ such that $\vee_{k\geq 0}\rho_k(X_i)=+\infty$. Thus $\crochet{C}^*(\rho_0)(X_i)=+\infty$ and we have proved that $\crochet{C}^*(\rho_0)$ upper-saturates $c_{i+1}$. This contradiction proves that $\crochet{C}^*(\rho_0)$ upper-saturates at least one constraint in $C$.
  \qed
\end{proof}

\section{Proof of Proposition~3.3}

\gdef\thesection{3}
\setcounter{theorem}{2}

\begin{proposition}
 The algorithm \textsf{CyclicSolve} returns $\crochet{C}^*(\rho_0)$ for any cyclic constraint system $\CS$ and for any valuation $\rho_0$.
\end{proposition}
\begin{proof}
    Let $\rho_1$, $\rho_2$, $\rho_3$, $\rho_4$, and $\rho_5$ be the value of $\rho$ just after the 1st, the 2sd, the 3th, the 4th and the 5th loops.\\

    Let us first show that if the $i_0$th iteration of the second loop does not modify the valuation $\rho$ then $\rho_2$ is a solution of $\CS$. Observe that the iterations $i_0,\ldots, n$ of the first loop and the iterations $1,\ldots,i_0-1$ of the second loop provide a valuation $\rho$ such that $\rho(X_i)\geq \crochet{c_i}(\rho)(X_i)$ for any $i\not=i_0$. As the $i_0$th iteration of the second loop does not modify $\rho$ we deduce that $\rho(X_{i_0})\geq \crochet{c_{i_0}}(X_{i_0})$. Therefore $\rho$ is a solution. We deduce that $\rho$ remains unchanged during the remaining iterations $i_0,\ldots,n$ of the second loop. Thus $\rho_2$ is a solution of $\CS$. 

    Assume that $\rho_2$ is not a pre-solution of $\CS$. There exists $i_0$ such that $\crochet{c_{i_0}}(\rho)(X_{i_0})\not\geq \rho(X_{i_0})$. We deduce that the value of $\rho$ has not been modified at the $i_0$th iteration of the 2sd loop. Thus, from the previous paragraph, $\rho_2$ is a solution. 

    Next, assume that a constraint $c_{i_0}$ is lower-saturated by $\rho_2$. Since after the $i_0$-iteration of the first loop we have $\rho(X_{i_0})\geq \crochet{c_{i_0}}(\bot)$, we deduce that the $i_0$th iteration of the second loop does not modify $\rho$. From the first paragraph we also deduce that $\rho_2$ is a solution of $\CS$.

    As the line 9 of the algorithm detects if $\rho_2$ is a solution, we can assume that $\rho_2$ is not a solution. From the two previous paragraph we deduce that $\rho_2$ is a pre-solution of $\CS$  and the constraints are not lower-saturated. From Lemma~\ref{lem:sunny} we deduce that $\crochet{C}^*(\rho_2)$ upper-saturates at least one constraint denoted by $c_{i_0}$. Observe that $\crochet{C}^*(\rho_2)=\crochet{C}^*(\rho_0)$.\\

    Let us show that $\crochet{c}(\crochet{C}^*(\rho_0))=\crochet{C}^*(\rho_0)$ for any constraint $c\in C$. Since $\rho_2$ is a pre-solution we get $\crochet{C}(\crochet{C}^*(\rho_2))=\crochet{C}^*(\rho_2)$. Moreover, as the output variables of two distinct constraints are distinct, we deduce that $\crochet{c}(\crochet{C}^*(\rho_2))=\crochet{C}^*(\rho_2)$ for any constraint $c\in C$. As $\crochet{C}^*(\rho_0)=\crochet{C}^*(\rho_2)$ we get the property.

    We deduce that the valuation $\rho'=\rho\wedge \crochet{c}(\rho)$ satisfies $\crochet{C}^*(\rho_0)\leq \rho'$ for any valuation $\rho$ such that $\crochet{C}^*(\rho_0)\leq \rho$ and for any constraint $c\in C$.\\

  After the 3th loop of the algorithm, we have $\crochet{C}^*(\rho_0)\leq \rho_3$, the previous paragraph proves that $\crochet{C}^*(\rho_0)\leq \rho$ is an invariant of the remaining of the program. Observe that at the $i_0$-th iteration of the 4th loop, we have $\rho(X_{i_0})=\crochet{C}^*(\rho_0)(X_{i_0})$. Thanks to the remaining iterations $i_0+1, \ldots, n$ of the 4th loop and the first iterations $1,\ldots i_0-1$ of the 5th loop, we get $\rho(X_{i})=\crochet{C}^*(\rho_0)(X_{i})$ for any $i$ since $\crochet{c}(\crochet{C}^*(\rho_0))=\crochet{C}^*(\rho_0)$ for any constraint $c\in C$. Thus at this point of the execution we have $\rho=\crochet{C}^*(\rho_0)$. Observe that $\rho$ is unchanged during the remaining iterations of the 5th loop. Thus, the algorithm returns $\crochet{C}^*(\rho_0)$.
\qed
\end{proof}

\section{Proof of Proposition~3.4}

\gdef\thesection{3}
\setcounter{theorem}{3}

\begin{proposition}
  The algorithm \textsf{SolveBI} returns the least solution $\crochet{C}^*(\rho_0)$ of a bounded-increasing constraint system $\CS$ greater than a valuation $\rho_0$. Moreover, the number of times the while loop is executed is bounded by one plus the number of constraints that are upper-saturated for $\crochet{C}^*(\rho_0)$ but not for $\rho_0$.
\end{proposition}
\begin{proof}
  Note that $\lambda$ is a partially defined function from $\X$ to $C$. At the beginning of the while loop this function is empty. Then, it is updated when the algorithm replaces a valuation $\rho$ by $\rho\vee \crochet{c}(\rho)$. Denoting by $X$ the output variable of $c$, the value $\lambda(X)$ becomes equal to $c$. That means $\lambda$ keeps in memory the last constraint that have modified a variable. Observe also that initially $\rho = \rho_0 \vee \crochet{C}(\rho_0)$. Thus, if during the execution of the algorithm $\rho(X)$ is updates by a constraint $c$ then necessary $c$ is not lower-saturated. That means if $\lambda(X)$ is defined then $c=\lambda(X)$ is not lower-saturated.\\

  Let $\rho_0$ and $\rho_1$ be the values of $\rho$ respectively before and after the execution of the first two nested loops (line 5-9) and let $\rho_2$ be the value of $\rho$ after the execution of line~14.\\

  We are going to prove that if the sets of upper-saturated constraints for $\rho_0$ and $\rho_1$ are equal and if there does not exist a cycle satisfying the condition given line 10, then $\rho_1$ is a solution of $\CS$. Let us consider the subset set of constraints $C'=\{\lambda(X);\; X\in\X\}$ and let us consider the graph $G'$ associated to the constraint system $(\X,C')$. We construct the graph $G_1$ obtained from $G$ by keeping only the transitions $(X,c)$ if $c=\lambda(X)$ and the transitions $(X_i,c)$. Observe that $G_1$ is acyclic. Thus, there exists an enumeration $c_1,\ldots, c_m$ of the set of constraints $C'$ such that if there exists a path from $c_{i_1}$ to $c_{i_2}$ in $G_1$ then $i_1\leq i_2$. Let us denote by $X_i$ the output variable of $c_i$.
  
  Let us prove by induction over $i$ that for any $j\in\finiteset{1}{i}$ we have $\rho_1(X_j)\leq \crochet{c_1\ldots c_j}(\rho_0)(X_j)$. The rank $i=0$ is immediate since in this case $\finiteset{1}{i}$ is empty. Let us assume that rank $i-1<n$ is true and let us prove the rank $i$. Since $\lambda(X_i)=c_i$ we deduce that the valuation $\rho_1(X_i)$ has been modified thanks to $c_i$. Thus, denoting by $\rho$ the valuation in the algorithm just before this update, we deduce that $\rho_1(X_i)=\crochet{c_i}(\rho)(X_i)$ and $\rho_0\leq \rho\leq \rho_1$. Let us prove that $\rho(X_{i,j})\leq \crochet{c_1\ldots c_{i-1}}(\rho_0)(X_{i,j})$ for any input variable $X_{i,j}$ of $c_i$. Observe that if $X_{i,j}\in \X'$ then $\rho_1(X_{i,j})=\rho(X_{i,j})=\rho_0(X_{i,j})$ by construction of $\lambda$ and in particular $\rho(X_{i,j})\leq\crochet{c_1\ldots c_{i-1}}(\rho_0)(X_{i,j})$ since $c_1,\ldots,c_{i-1}$ do not modify the variable $X_{i,j}$. Otherwise, if $X_{i,j}\in\X'$, there exists $i'<i$ satisfying $X_{i,j}=X_{i'}$. By induction hypothesis, we have $\rho_1(X_{i'})\leq \crochet{c_1\ldots c_{i'}}(\rho_0)(X_{i'})$. Since $c_1,\ldots, c_m$ have distinct output variables, we deduce that $\crochet{c_1\ldots c_{i'}}(\rho_0)(X_{i'})=\crochet{c_1\ldots c_{i-1}}(\rho_0)(X_{i'})$. Thus $\rho_1(X_{i,j})\leq \crochet{c_1\ldots c_{i-1}}(\rho_0)(X_{i,j})$ and from $\rho\leq \rho_1$, we get $\rho(X_{i,j})\leq \crochet{c_1\ldots c_{i-1}}(\rho_0)(X_{i,j})$ for any input variable $X_{i,j}$. Therefore $\crochet{c_i}(\rho)(X_i)\leq \crochet{c_1\ldots c_{i}}(\rho_0)(X_{i})$. From $\rho_1(X_i)=\crochet{c_i}(\rho)(X_i)$, we get $\rho_1(X_i)\leq \crochet{c_1\ldots c_i}(\rho_0)(X_i)$ and we have proved the induction.
  
  We deduce the relation $\rho_1\leq \crochet{c_1\ldots c_m}(\rho_0)$ since $c_1$, ..., $c_m$ have distinct output variables. Observe that after the first execution of the loop 6-9, we get $\rho\geq \crochet{c_1}(\rho_0)$, after the second $\rho\geq \crochet{c_1.c_2}(\rho_0)$. By induction, after $m$ executions we get $\rho\geq \crochet{c_1\ldots c_m}(\rho_0)\geq \rho_1$. Since $m\leq |C|$, this loop is executed at least one more time. Note that after this execution, we have $\rho\geq \crochet{c}(\rho_1)$ for any $c\in C$. Since $\rho_1\geq \rho$, we have proved that $\rho_1\geq \crochet{C}(\rho_1)$. Therefore $\rho_1$ is a solution of $C$.\\
  
  Next, assume that there exists a cycle $X_0\rightarrow c_1\rightarrow X_1\cdots c_n\rightarrow X_n=X_0$ that satisfies $c_i=\lambda(X_i)$. From the first paragraph we deduce that $c_1,\ldots, c_n$ are not lower-saturated.  Let us prove that there exists a constraint upper-saturated for $\rho_2$ that is not upper-saturated for $\rho_0$. Naturally, if there exists a constraints upper-saturated from $\rho_1$ that is not upper-saturated for $\rho_0$, since $\rho_1\leq \rho_2$, we are done. Thus, we can assume that the constraints $c_1$, ..., $c_n$ are not upper-saturated for $\rho_1$. By definition of $\lambda$, we get $\rho_i(X_i)\leq \crochet{c_i}(\rho)$. Thus $\rho'$ is a pre-solution of $\CS'$. Let $X_i$ be the last variable amongst $X_0,\ldots,X_{n-1}$ that have been updated. Since $c_{i+1}$ is not upper-saturated and not lower-saturated for $\rho_1$ and since the value of $X_i$ has increased when this last update appeared, we deduce that $\rho'(X_{i+1})\not\geq \crochet{c_{i+1}}(\rho')(X_{i+1})$. Thus $\rho'$ is not a solution and from lemma \ref{lem:sunny} we deduce that $\rho''$ upper-saturates at least one constraint $c_i$. Thus $\rho_2$ upper-saturates a constraints that is not upper-saturated by $\rho_1$.\\

  Finally, note that each time the while loop is executed at least one bounded-increasing constraint becomes upper-saturated. As every upper-saturated constraint remains upper-saturated, we are done.
\qed
\end{proof}

} 

\end{document}